\documentclass[11pt, reqno]{amsart}
\usepackage{amsmath, amsthm, amscd, amsfonts, amssymb, graphicx, color}
\usepackage{bbm,mathtools}
\usepackage{booktabs}
\usepackage[bookmarksnumbered, colorlinks, plainpages]{hyperref}
\usepackage{enumitem}
\usepackage{mathrsfs}
\hypersetup{colorlinks=true,linkcolor=red, anchorcolor=green, citecolor=cyan, urlcolor=red, filecolor=magenta, pdftoolbar=true}

\newtheorem{theorem}{Theorem}[section]
\newtheorem{lemma}[theorem]{Lemma}

\newtheorem{corollary}[theorem]{Corollary}
\theoremstyle{definition}
\newtheorem{definition}[theorem]{Definition}

\theoremstyle{remark}
\newtheorem{remark}[theorem]{Remark}

\newcommand{\M}{\mathbb{M}}

\begin{document}

\title{Kubo-Ando Means and Rigidity of Quantum Positivity Cones}

\author{Mohsen Kian}
\address{Mohsen Kian:  Department of Mathematics, University of Bojnord, P. O. Box 1339, Bojnord 94531, Iran}
\email{kian@ub.ac.ir }

 \begin{abstract}
We investigate the stability of quantum positivity cones under nonlinear operator means. Specifically, we examine how Kubo--Ando means interact with the separable, positive partial transpose (PPT), and Schmidt-number cones. By analyzing the curvature of operator monotone  functions at the identity, we give  a strict rigidity phenomenon: weighted arithmetic means are the only Kubo--Ando means that preserve the separable cone in all dimensions. We show  that the strictly positive curvature of any non-arithmetic mean explicitly forces a violation of the PPT condition, even in the foundational two-qubit setting, and can strictly increase the Schmidt number of the resulting operator. Finally, using  the Choi--Jamio{\l}kowski correspondence, we translate these geometric obstructions to the map-theoretic setting, concluding that convex mixing is the uniquely permissible Kubo--Ando operation for preserving entanglement-breaking quantum channels.
\end{abstract}

\maketitle

\section{Introduction}

Kubo--Ando means provide the standard axiomatic framework for binary
means of positive operators.  Introduced by Kubo and Ando
\cite{KuboAndo1980}, these means are characterized by natural order,
continuity, and transformer-inequality properties, and are equivalently
described by operator monotone representing functions.  Thus they form a
distinguished class of nonlinear operations on the positive cone.  Many
familiar means, such as the arithmetic, harmonic, geometric, and
logarithmic means, arise in this framework.

The purpose of this paper is to study how Kubo--Ando means interact with
cones coming from quantum information theory.  The relevant cones are not
only the full positive semidefinite cone, but also finer cones determined
by entanglement constraints.  The separable cone consists of positive
operators that are finite sums of tensor products of positive operators.
Its trace-one slice is the set of separable bipartite quantum states.
The PPT cone is defined by positivity under partial transposition, a
condition introduced by Peres \cite{Peres1996}.  The fact that PPT is
equivalent to separability in dimensions \(2\otimes2\) and \(2\otimes3\)
is the Peres--Horodecki criterion \cite{Horodecki1996}.  More generally,
the Schmidt-number cones, introduced by Terhal and Horodecki
\cite{TerhalHorodecki2000}, interpolate between the separable cone and
the full positive semidefinite cone.  Through the
Choi--Jamio{\l}kowski correspondence
\cite{Jamiolkowski1972,Choi1975}, these structures are closely connected
with completely positive maps and entanglement-breaking channels
\cite{HorodeckiShorRuskai2003}.

The interaction between operator means and positivity has been studied
extensively \cite{KMS}, but much less is known about their behavior on
entanglement-related cones.  In this paper we study how Kubo--Ando means
act on separable, PPT, and Schmidt-number cones.  Our main result shows
that weighted arithmetic means are distinguished by strong preservation
properties for these cones, whereas non-arithmetic means fail to preserve
them even in the two-qubit setting.  Via the Choi--Jamio{\l}kowski
correspondence, we also obtain parallel statements for Choi matrices and
entanglement-breaking maps.

\section{Preliminaries}
Let \(\M_m^+\) denote the cone of positive semidefinite \(m\times m\)
matrices. The separable cone in \(\M_m\otimes \M_n\) is
\[
\mathcal S_1^{m,n}
=
\left\{
\sum_{j=1}^r A_j\otimes B_j:
r<\infty,\ A_j\in \M_m^+,\ B_j\in \M_n^+
\right\}.
\]
Equivalently, \(\mathcal S_1^{m,n}\) is the conic hull of the rank-one
projections \(zz^*\) with \(z=x\otimes y\) a product vector:
\[
        \mathcal S_1^{m,n}
        =
        \operatorname{cone}
        \{
        (xx^*)\otimes(yy^*):
        x\in\mathbb C^m,\ y\in\mathbb C^n
        \}.
\]
 Its trace-one
slice is precisely the set of separable bipartite quantum states.
Let \(T_n:\M_n\to \M_n\) denote the transpose map with respect to the
standard basis. For \(X\in \M_m\otimes \M_n\), we write
\[
X^\Gamma=(\operatorname{id}_m\otimes T_n)(X)
\]
and call \(X^\Gamma\) the partial transpose of \(X\). The PPT cone is
defined by
\[
\mathcal {PPT}^{m,n}
=
\{X\in \M_m\otimes \M_n: X\geq 0 \text{ and } X^\Gamma\geq 0\}.
\]
Clearly, every separable operator is PPT, i.e.,  $\mathcal S_1^{m,n}\subseteq \mathcal {PPT}^{m,n}$.

For \(z\in \mathbb C^m\otimes \mathbb C^n\), let
\(\operatorname{SR}(z)\) denote the Schmidt rank of \(z\), i.e. the
least integer \(r\) such that
\[
z=\sum_{i=1}^r x_i\otimes y_i
\]
for some \(x_i\in\mathbb C^m\) and \(y_i\in\mathbb C^n\). For
\(1\leq k\leq \min(m,n)\), the Schmidt-number cone is defined by
\[
\mathcal S_k^{m,n}
=
\operatorname{cone}
\left\{
zz^*:
z\in\mathbb C^m\otimes\mathbb C^n,\
\operatorname{SR}(z)\leq k
\right\}.
\]
Equivalently, \(X\in\mathcal S_k^{m,n}\) if and only if
\[
X=\sum_{j=1}^r z_jz_j^*
\]
for finitely many vectors \(z_j\) satisfying
\(\operatorname{SR}(z_j)\leq k\). In particular,
\(\mathcal S_1^{m,n}\) is the separable cone, while
$
\mathcal S_{\min(m,n)}^{m,n}
$
is the whole positive semidefinite cone.
Thus the cones \(\mathcal S_k^{m,n}\) interpolate between the separable
cone and the full positive semidefinite cone.

We shall use the Choi--Jamio{\l}kowski correspondence in the following
form.  If \(\Phi:\M_m\to \M_n\) is a linear map, its Choi matrix is
\[
    C_\Phi
    =
    \sum_{i,j=1}^m E_{ij}\otimes \Phi(E_{ij})
    \in \M_m\otimes \M_n,
\]
where \(\{E_{ij}\}_{i,j=1}^m\) denotes the standard matrix units in
\(\M_m\).  The map \(\Phi\mapsto C_\Phi\) is a linear isomorphism between
linear maps \(\M_m\to \M_n\) and matrices in \(\M_m\otimes \M_n\).  Under
this correspondence, \(\Phi\) is completely positive if and only if
\(C_\Phi\geq 0\).  Moreover, \(\Phi\) is trace-preserving if and only if
\[
    \operatorname{Tr}_2(C_\Phi)=I_m,
\]
where \(\operatorname{Tr}_2\) denotes the partial trace over the second
tensor factor.  Finally, \(\Phi\) is entanglement-breaking precisely when
\(C_\Phi\) is separable, that is,
\[
    C_\Phi\in\mathcal S_1^{m,n}.
\]

\section{Main result}
In this section we establish the main rigidity result for Kubo--Ando
means with respect to separability and related entanglement cones.  We
begin by introducing the quantity that will play a central role in the
proofs.

\begin{definition}[Curvature of a Kubo--Ando mean]
Let \(\sigma\) be a Kubo--Ando mean with representing function \(f\).
We define the curvature of \(\sigma\) by
\[
        \kappa_\sigma=-f''(1),
\]
whenever \(f\) is twice differentiable at \(1\).
\end{definition}

If \(\sigma\) is a non-affine Kubo--Ando mean with representing
function \(f\), then
\[
    \kappa_\sigma := -f''(1) >0.
\]
For the arithmetic mean, one has \(\kappa_\sigma=0\). Noting that   \(f:(0,\infty)\to(0,\infty)\) is operator concave \cite{FMPS} and satisfies
\(f(1)=1\), we know      \(f\) is concave.  Hence
\[
    f''(1)\le 0,
\]
whenever the second derivative is taken in the ordinary sense.

Moreover, equality \(f''(1)=0\) can occur only when \(f\) is affine.
Indeed,  operator monotone function   \(f\) can be written in the form
\[
    f(t)=a+bt+\int_{(0,\infty)}
    \frac{(1+\lambda)t}{t+\lambda}\,\mathrm{d}\mu(\lambda),
\]
where \(a,b\ge 0\) and \(\mu\) is a positive finite measure.  Differentiating
twice gives
\[
    f''(t)
    =
    -2\int_{(0,\infty)}
    \frac{(1+\lambda)\lambda}{(t+\lambda)^3}\,\mathrm{d}\mu(\lambda),
\]
whence
\[
    f''(1)
    =
    -2\int_{(0,\infty)}
    \frac{\lambda}{(1+\lambda)^2}\,\mathrm{d}\mu(\lambda)
    \le 0.
\]
If \(f\) is not affine, then \(\mu\neq 0\), and since
\[
    \frac{\lambda}{(1+\lambda)^2}>0
    \qquad (\lambda>0),
\]
we get
\[
    f''(1)<0.
\]
Consequently,
\[
    \kappa_\sigma:=-f''(1)>0.
\]

\begin{remark}
For several standard Kubo--Ando means, the curvature is as follows:
\[
\begin{array}{c|c|c}
\text{Mean} & f(x) & \kappa_\sigma=-f''(1) \\ \hline
\text{weighted arithmetic} &
(1-\alpha)+\alpha x &
0 \\[0.4em]

\text{weighted geometric} &
x^\alpha &
\alpha(1-\alpha) \\[0.4em]

\text{weighted harmonic} &
\dfrac{x}{(1-\alpha)x+\alpha} &
2\lambda(1-\alpha) \\[0.8em]

\text{weighted logarithmic }
&
\dfrac{\alpha-1}{\alpha}\,
\dfrac{x^\alpha-1}{x^{\alpha-1}-1}
&
\dfrac{1-\alpha+\alpha^2}{6}
\\[1.5ex]
\text{dual weighted logarithmic }
&
\dfrac{\alpha}{\alpha-1}\,
\dfrac{x(x^{\alpha-1}-1)}{x^\alpha-1}
&
\dfrac{1+\alpha-\alpha^2}{6}
\end{array}
\]
Here \(0\leq\alpha\leq1\). Thus the weighted arithmetic means are exactly
the zero-curvature means, while every non-arithmetic example in the table
has strictly positive curvature.
\end{remark}

We first establish the fundamental two-qubit construction underlying the
main results.
\begin{theorem}\label{thm-main1}
Let \(\sigma\) be a Kubo--Ando mean with representing function \(f\). Then the following are equivalent:
\begin{enumerate}
   \item \(\mathcal S_1^{2,2}\) is closed under \(\sigma\).
    \item \(\sigma\) is a weighted arithmetic mean.
 \end{enumerate}
In particular,  if \(\sigma\) is not arithmetic, then there exist positive
definite operators $ A,B\in\mathcal S_1^{2,2}$ such that $ A\sigma B\notin\mathcal S_1^{2,2}$. Moreover, the operators \(A\) and \(B\) may be chosen from $    \operatorname{int}\mathcal S_1^{2,2}$.
 \end{theorem}
\begin{proof}
If \(\sigma\) is a weighted arithmetic mean, say
\[
        A\sigma B=(1-t)A+tB
        \qquad (0\leq t\leq1),
\]
then \(\mathcal S_1^{2,2}\) is closed under \(\sigma\), since
\(\mathcal S_1^{2,2}\) is a convex cone.

Conversely, assume that \(\sigma\) is not arithmetic. Then
\[
        \kappa_\sigma=-f''(1)>0.
\]
We construct two separable two-qubit operators whose mean is not PPT.
Let $\{e_1,e_2\}$ be the standard basis of $\mathbb{C}^2$. Put
\[
        u=e_1\otimes e_1,\qquad
        v=e_2\otimes e_2,
\]
and
\[
        s=\frac{e_1\otimes e_2+e_2\otimes e_1}{\sqrt2},
        \qquad
        a=\frac{e_1\otimes e_2-e_2\otimes e_1}{\sqrt2}.
\]
so that $\{u,v,s,a\}$ is an  orthonormal basis of \(\mathbb C^2\otimes\mathbb C^2\).
For \(0<\varepsilon<1\), define
\[
\begin{aligned}
        A_\varepsilon
        &=
        \frac{1+\varepsilon}{2}(uu^*+vv^*)
        +
        \frac{1-\varepsilon}{2}ss^*
        +
        \frac{3+\varepsilon}{2}aa^*,
        \\[0.4em]
        B_\varepsilon
        &=
        \frac{1-\varepsilon}{2}(uu^*+vv^*)
        +
        \frac{1+\varepsilon}{2}ss^*
        +
        \frac{3-\varepsilon}{2}aa^* .
\end{aligned}
\]
Clearly,  \(A_\varepsilon\) and \(B_\varepsilon\) are positive definite. They both are of the form
\[
        X=\alpha(uu^*+vv^*)+\beta ss^*+\gamma aa^*,
\]
for which the partial transpose $X^\Gamma = (\operatorname{id}_2\otimes T)(X)$
  can be represented in the standard computational basis as the block-diagonal matrix:
\[
        X^\Gamma = \begin{pmatrix} \alpha & 0 & 0 & \frac{\beta-\gamma}{2} \\ 0 & \frac{\beta+\gamma}{2} & 0 & 0 \\ 0 & 0 & \frac{\beta+\gamma}{2} & 0 \\ \frac{\beta-\gamma}{2} & 0 & 0 & \alpha \end{pmatrix}.
\]
Hence the eigenvalues of \(X^\Gamma\) are
\[
        \frac{\beta+\gamma}{2},\quad
        \frac{\beta+\gamma}{2},\quad
        \alpha+\frac{\beta-\gamma}{2},\quad
        \alpha-\frac{\beta-\gamma}{2}.
\]

For the operator $A_\varepsilon$, we have $\alpha = \frac{1+\varepsilon}{2}$, $\beta = \frac{1-\varepsilon}{2}$, and $\gamma = \frac{3+\varepsilon}{2}$. Substituting these into the eigenvalue formulas gives us:
\begin{align*}
    \lambda_{1,2} &= \frac{\beta+\gamma}{2} = 1 > 0, \\
    \lambda_3 &= \alpha + \frac{\beta-\gamma}{2} = \frac{1+\varepsilon}{2} + \frac{-1-\varepsilon}{2} = 0, \\
    \lambda_4 &= \alpha - \frac{\beta-\gamma}{2} = \frac{1+\varepsilon}{2} - \frac{-1-\varepsilon}{2} = 1+\varepsilon > 0.
\end{align*}
Thus, all eigenvalues of $A_\varepsilon^\Gamma$ are non-negative.

Similarly, for $B_\varepsilon$, we have $\alpha = \frac{1-\varepsilon}{2}$, $\beta = \frac{1+\varepsilon}{2}$, and $\gamma = \frac{3-\varepsilon}{2}$. The eigenvalues are:
\begin{align*}
    \lambda_{1,2} &= \frac{\beta+\gamma}{2} = 1 > 0, \\
    \lambda_3 &= \alpha + \frac{\beta-\gamma}{2} = \frac{1-\varepsilon}{2} + \frac{-1+\varepsilon}{2} = 0, \\
    \lambda_4 &= \alpha - \frac{\beta-\gamma}{2} = \frac{1-\varepsilon}{2} - \frac{-1+\varepsilon}{2} = 1-\varepsilon > 0.
\end{align*}
Since $0 < \varepsilon < 1$, all eigenvalues of $B_\varepsilon^\Gamma$ are also non-negative.

Since both $A_\varepsilon$ and $B_\varepsilon$ are positive definite and their partial transposes are positive semidefinite, we have
\[
        A_\varepsilon,B_\varepsilon \in \mathcal {PPT}^{2,2}.
\]
It follows from  the Peres--Horodecki criterion \cite{Horodecki1996,Peres1996} in dimension \(2\otimes2\) that $\mathcal {PPT}^{2,2} = \mathcal S_1^{2,2}$. Consequently,
\[
        A_\varepsilon,B_\varepsilon \in \mathcal S_1^{2,2}.
\]

Since \(A_\varepsilon\) and \(B_\varepsilon\) are diagonal in the same
orthonormal basis, their Kubo--Ando mean is also diagonal in this basis:
\[
        A_\varepsilon\sigma B_\varepsilon
        =
        \alpha_\varepsilon(uu^*+vv^*)
        +
        \beta_\varepsilon ss^*
        +
        \gamma_\varepsilon aa^*,
\]
where
 \begin{align}\label{const}
        \alpha_\varepsilon
       =
        \frac{1+\varepsilon}{2}
        f\!\left(\frac{1-\varepsilon}{1+\varepsilon}\right),
        \quad
        \beta_\varepsilon
        =
        \frac{1-\varepsilon}{2}
        f\!\left(\frac{1+\varepsilon}{1-\varepsilon}\right),
        \quad
        \gamma_\varepsilon
        =
        \frac{3+\varepsilon}{2}
        f\!\left(\frac{3-\varepsilon}{3+\varepsilon}\right).
\end{align}

Clearly, $A_\varepsilon\sigma B_\varepsilon$ is positive definite. We show that its partial transpose has a negative eigenvalue.

Because $A_\varepsilon\sigma B_\varepsilon$ is exactly of the form of $A_\varepsilon$ and $B_\varepsilon$, the spectrum of its partial transpose can be computed  by the same eigenvalue formulas. Specifically, one of the eigenvalues of $(\operatorname{id}_2\otimes T)(A_\varepsilon\sigma B_\varepsilon)$ is:
\[
        \lambda_3 = \alpha_\varepsilon + \frac{\beta_\varepsilon}{2} - \frac{\gamma_\varepsilon}{2}.
\]
Using \eqref{const} and applying  the Taylor expansion
\[
        f(1+h)
        =
        1+f'(1)h-\frac{\kappa_\sigma}{2}h^2+O(h^3),
\]
a direct computation yields
\[
        \lambda_3
        =
        \alpha_\varepsilon + \frac{\beta_\varepsilon}{2} - \frac{\gamma_\varepsilon}{2}
        =
        -\frac43\kappa_\sigma\varepsilon^2 + O(\varepsilon^3),
\]
where $\kappa_\sigma$ is the curvature of $\sigma$.  Since \(\kappa_\sigma>0\), this eigenvalue is strictly negative for all sufficiently small \(\varepsilon>0\). Therefore, the partial transpose of $A_\varepsilon\sigma B_\varepsilon$ has a negative eigenvalue, which means
\[
        (\operatorname{id}_2\otimes T)(A_\varepsilon\sigma B_\varepsilon) \ngeq 0,
\]
and hence
\[
        A_\varepsilon\sigma B_\varepsilon \notin \mathcal S_1^{2,2}.
\]

Now, we prove that the input operators  can be chosen from the interior of   separable cone. Chose
\[
        0<c<\frac43\kappa_\sigma
\]
and set
\[
        A_{\varepsilon,c}=A_\varepsilon+c\varepsilon^2 I,
        \qquad
        B_{\varepsilon,c}=B_\varepsilon+c\varepsilon^2 I.
\]
Then $A_{\varepsilon,c}$ and $B_{\varepsilon,c}$ are  positive definite, and their partial transposes satisfy
\[
        (\operatorname{id}_2\otimes T)(A_{\varepsilon,c}) = (\operatorname{id}_2\otimes T)(A_\varepsilon) + c\varepsilon^2 I > 0,
\]
and
\[
        (\operatorname{id}_2\otimes T)(B_{\varepsilon,c}) = (\operatorname{id}_2\otimes T)(B_\varepsilon) + c\varepsilon^2 I > 0.
\]
Thus
\[
        A_{\varepsilon,c},B_{\varepsilon,c}
        \in
        \operatorname{int}\mathcal {PPT}^{2,2}
        =
        \operatorname{int}\mathcal S_1^{2,2}.
\]
Since $A_{\varepsilon,c}$ and $B_{\varepsilon,c}$ are diagonal in the same orthonormal basis $\{u, v, s, a\}$ as before, their Kubo--Ando mean $A_{\varepsilon,c}\sigma B_{\varepsilon,c}$ is also diagonal in this basis. Moreover, we have
\[
        A_{\varepsilon,c}\sigma B_{\varepsilon,c}
        =
        A_\varepsilon\sigma B_\varepsilon + c\varepsilon^2 I + O(\varepsilon^4).
\]
Because adding $c\varepsilon^2 I$ simply shifts the spectrum of the partial transpose by $c\varepsilon^2$, the corresponding smallest eigenvalue of $(\operatorname{id}_2\otimes T)(A_{\varepsilon,c}\sigma B_{\varepsilon,c})$ becomes:
\begin{align*}
        \tilde{\lambda}_3
        &=
        \lambda_3 + c\varepsilon^2 + O(\varepsilon^4) \\
        &=
        \left(c-\frac43\kappa_\sigma\right)\varepsilon^2
        +
        O(\varepsilon^3).
\end{align*}
By the choice of \(c\), this eigenvalue is strictly negative for all sufficiently small \(\varepsilon>0\). Therefore, the partial transpose has a negative eigenvalue, which means
\[
        A_{\varepsilon,c}\sigma B_{\varepsilon,c}
        \notin
        \mathcal {PPT}^{2,2} = \mathcal S_1^{2,2}.
\]
This proves the converse implication and completes the proof.
\end{proof}

We now extend the preceding two-qubit obstruction to arbitrary
dimensions.

\begin{theorem}
\label{thm-main2}
Let \(m,n\geq2\), and let \(\sigma\) be a Kubo--Ando mean. Then the
following are equivalent:
\begin{enumerate}
     \item \(\mathcal S_1^{m,n}\) is closed under \(\sigma\).
    \item \(\sigma\) is a weighted arithmetic mean.
   \end{enumerate}
\end{theorem}

\begin{proof}
 If \(\sigma\) is a weighted arithmetic mean, then    \(\mathcal S_1^{m,n}\) is closed  under \(\sigma\), because it is a convex cone.

Conversely, suppose that \(\sigma\) is not arithmetic. By
Theorem~\ref{thm-main1}, there exist
$     A,B\in\mathcal S_1^{2,2}$
such that $ A\sigma B\notin\mathcal S_1^{2,2}$.

Suppose that
\[
        V:\mathbb C^2\to\mathbb C^m,
        \qquad
        W:\mathbb C^2\to\mathbb C^n
\]
are  the canonical isometries onto the spans of the first two standard basis
vectors. We define
\[
        \widetilde A=(V\otimes W)A(V\otimes W)^*,
        \qquad
        \widetilde B=(V\otimes W)B(V\otimes W)^*.
\]
Since local isometries preserve separability,
\[
        \widetilde A,\widetilde B\in\mathcal S_1^{m,n}.
\]
On the other hand, the congruence invariance of Kubo--Ando implies that,
\begin{align}\label{atilde}
        \widetilde A\sigma\widetilde B
        =
        (V\otimes W)(A\sigma B)(V\otimes W)^*.
\end{align}

Assume, for contradiction, that
\[
        \widetilde A\sigma\widetilde B
        \in\mathcal S_1^{m,n}.
\]
Since $V\otimes W$ is an  isometry, it follows form \eqref{atilde} that
\[
        A\sigma B
        =
        (V^*\otimes W^*)
        (\widetilde A\sigma\widetilde B)
        (V\otimes W)
        \in\mathcal S_1^{2,2},
\]
which contradicts the choice of \(A\) and \(B\). Hence
\[
        \widetilde A\sigma\widetilde B
        \notin\mathcal S_1^{m,n}.
\]
Therefore \(\mathcal S_1^{m,n}\) is not closed under \(\sigma\).
\end{proof}

We next show that the failure of preservation extends beyond
separability.  More precisely, non-arithmetic Kubo--Ando means can
increase Schmidt number and therefore fail to preserve the cones
\(\mathcal S_r^{m,n}\).
\begin{theorem}\label{thm-main3}
Let \(m,n\geq2\), and let \(\sigma\) be a Kubo--Ando mean. Let
\(\mathcal C\subseteq \M_m\otimes \M_n\) be a convex cone satisfying
\[
        \mathcal S_1^{m,n}
        \subseteq
        \mathcal C
        \subseteq
        \mathcal {PPT}^{m,n}.
\]
Then the following are equivalent:
\begin{enumerate}
\item \(\mathcal C\) is closed under \(\sigma\).
    \item \(\sigma\) is a weighted arithmetic mean.
    \end{enumerate}
\end{theorem}
\begin{proof}
Trivially, every convex cone is closed under the  weighted arithmetic mean.

Conversely, suppose that \(\sigma\) is not arithmetic. By
Theorem~\ref{thm-main1}, there exist
$  A,B\in\operatorname{int}\mathcal S_1^{2,2}$
such that
\[
        (\operatorname{id}_2\otimes T)(A\sigma B)\ngeq0.
\]
By embedding this two-qubit construction into
        $\mathbb C^m\otimes\mathbb C^n$
 by local isometries, as in the proof of
Theorem~\ref{thm-main2}, we obtain
operators
        $\widetilde A,\widetilde B\in\mathcal S_1^{m,n}$
such that
\[
        (\operatorname{id}_m\otimes T)(\widetilde A\sigma\widetilde B)
        \ngeq0.
\]
Hence
\[
        \widetilde A\sigma\widetilde B
        \notin
        \mathcal {PPT}^{m,n}.
\]
Since
\[
        \mathcal C\subseteq \mathcal {PPT}^{m,n},
\]
we have
\[
        \widetilde A\sigma\widetilde B\notin\mathcal C.
\]
On the other hand,
\[
        \widetilde A,\widetilde B\in\mathcal S_1^{m,n}
        \subseteq \mathcal C.
\]
Thus \(\mathcal C\) is not closed under \(\sigma\).
\end{proof}

This structural limitation is particularly relevant when considering concrete geometric families that interpolate between the separable and PPT boundaries. For instance, a natural class of candidates for such intermediate strata is given by the families of nested entanglement cones arising from rank-constrained $C^*$-convex hulls, see \cite{kian-rc2025}. Theorem~\ref{thm-main3} implies that despite the rich facial and algebraic structure possessed by the cones within this hierarchy, none of them can maintain closure under a non-affine Kubo–Ando operation.

\begin{lemma} \label{lem-schmidt-number}
Let $\psi$ be a vector in a tensor product space with Schmidt rank $r$, and let $Y \geq 0$ be a non-zero positive semidefinite operator on a tensor product space. Then
\[
        \operatorname{SN}(\psi\psi^*\otimes Y)
        =
        r\,\operatorname{SN}(Y).
\]
\end{lemma}

\begin{proof}
Let $k = \operatorname{SN}(Y)$. By definition, there exists a decomposition $Y = \sum_j \phi_j \phi_j^*$ where each vector $\phi_j$ has Schmidt rank $\operatorname{SR}(\phi_j) \leq k$. Tensoring with $\psi\psi^*$ yields the decomposition
\[
        \psi\psi^* \otimes Y = \sum_j (\psi \otimes \phi_j)(\psi \otimes \phi_j)^*.
\]
Under the bipartite grouping
\[
(\mathcal H_1\otimes\mathcal H_2)\otimes
(\mathcal K_1\otimes\mathcal K_2)
\cong
(\mathcal H_1\otimes\mathcal K_1)\otimes
(\mathcal H_2\otimes\mathcal K_2),
\]
Schmidt rank is multiplicative on simple tensor products. Therefore, $\operatorname{SR}(\psi \otimes \phi_j) = \operatorname{SR}(\psi)\operatorname{SR}(\phi_j) \leq rk$. This shows that $\operatorname{SN}(\psi\psi^* \otimes Y) \leq r \operatorname{SN}(Y)$.

Conversely,   assume that $\operatorname{SN}(\psi\psi^* \otimes Y)=s$ and consider an optimal decomposition
\[
        \psi\psi^* \otimes Y = \sum_j \xi_j \xi_j^*,
\]
where $\operatorname{SR}(\xi_j) \leq s$ for all $j$. Because the range of this operator is contained entirely within the subspace spanned by $\psi$ in the first tensor factor, each vector in the decomposition must be of the form $\xi_j = \psi \otimes \eta_j$ for some vectors $\eta_j$. Hence
\[
        \operatorname{SR}(\xi_j) = \operatorname{SR}(\psi)\operatorname{SR}(\eta_j) = r \operatorname{SR}(\eta_j).
\]
Since we established $\operatorname{SR}(\xi_j) \leq s$, it follows that $\operatorname{SR}(\eta_j) \leq s/r$. Finally, evaluating the decomposition gives
\[
        \psi\psi^* \otimes Y = \sum_j (\psi\psi^* \otimes \eta_j \eta_j^*),
\]
which implies $Y = \sum_j \eta_j \eta_j^*$. Thus \(Y\) has a decomposition by vectors of Schmidt rank at most \(s/r\), meaning $\operatorname{SN}(Y) \leq s/r$.

Rearranging this yields $r \operatorname{SN}(Y) \leq s = \operatorname{SN}(\psi\psi^* \otimes Y)$. Combining this with the upper bound completes the proof.
\end{proof}

\begin{theorem}\label{thm-main4}
Let \(\sigma\) be a non-arithmetic Kubo--Ando mean. Let
\(r\geq1\), and suppose that
\[
        m,n\geq 2r.
\]
Then there exist operators
\[
        A,B\in \mathcal S_r^{m,n}
\]
such that
\[
        A\sigma B\notin \mathcal S_{2r-1}^{m,n}.
\]
In particular,
\[
        \operatorname{SN}(A\sigma B)\geq 2r.
\]
\end{theorem}

\begin{proof}
Since \(\sigma\) is not arithmetic, Theorem~\ref{thm-main1}
gives two operators $A_0,B_0\in\mathcal S_1^{2,2}$ such that $ A_0\sigma B_0\notin\mathcal S_1^{2,2}$.
In dimension \(2\otimes2\), this means $\operatorname{SN}(A_0\sigma B_0)=2$. Let $\{e_1,\dots,e_r\}$ be the standard orthonormal basis of $\mathbb C^r$ and put
\[
        \psi_r=\sum_{j=1}^r e_j\otimes e_j
        \in\mathbb C^r\otimes\mathbb C^r .
\]
Then $\operatorname{SR}(\psi_r)=r$. We set $P_r=\psi_r\psi_r^*$  and then consider the two operators $        P_r\otimes A_0$ and $P_r\otimes B_0$ on
\[
        (\mathbb C^r\otimes\mathbb C^2)
        \otimes
        (\mathbb C^r\otimes\mathbb C^2).
\]
 By applying the canonical unitary isomorphism that exchanges the second and third tensor factors, we naturally identify this space with
\[
        (\mathbb C^r\otimes\mathbb C^2)
        \otimes
        (\mathbb C^r\otimes\mathbb C^2).
\]
Under this canonical identification, we regard
\[
        P_r\otimes A_0,
        \qquad
        P_r\otimes B_0
\]
as bipartite operators on this rearranged space. Under the natural identifications $\mathbb C^r\otimes\mathbb C^2\cong\mathbb C^{2r}$, these become operators on
\[
        \mathbb C^{2r}\otimes\mathbb C^{2r}.
\]
Since \(A_0\) and \(B_0\) are separable,
\[
        \operatorname{SN}(A_0)\leq1,
        \qquad
        \operatorname{SN}(B_0)\leq1.
\]
Applying  Lemma~\ref{lem-schmidt-number}, we obtain
\[
        \operatorname{SN}(P_r\otimes A_0)=\operatorname{SN}(\psi\psi^*\otimes A_0)
        =
        \operatorname{SR}(\psi)\operatorname{SN}(A_0)\leq r.
\]
Similarly, $\operatorname{SN}(P_r\otimes B_0)\leq r$ and so
 \[
        P_r\otimes A_0,\;
        P_r\otimes B_0
        \in
        \mathcal S_r^{2r,2r}.
\]
Noting that Kubo--Ando means have a tensorial congruence property, we have
\[
        (P_r\otimes A_0)\sigma(P_r\otimes B_0)
        =
        P_r\otimes(A_0\sigma B_0).
\]
Applying the same standard tensoring property again gives
\[
\begin{aligned}
        \operatorname{SN}
        \bigl(
        (P_r\otimes A_0)\sigma(P_r\otimes B_0)
        \bigr)
        &=
        \operatorname{SN}\bigl(P_r\otimes(A_0\sigma B_0)\bigr)
        \\
        &=
        r\,\operatorname{SN}(A_0\sigma B_0)
        \\
        &=2r .
\end{aligned}
\]
Thus
\[
        (P_r\otimes A_0)\sigma(P_r\otimes B_0)
        \notin
        \mathcal S_{2r-1}^{2r,2r}.
\]

Finally, if \(m,n\geq2r\), we use the embedding
\[
        \mathbb C^{2r}\otimes\mathbb C^{2r}
        \subseteq
        \mathbb C^m\otimes\mathbb C^n
\]
by local isometries, exactly as in the proof of
Theorem~\ref{thm-main2}. Local
isometries preserve Schmidt number and preserve Kubo--Ando means by
congruence invariance. Therefore the same construction gives
\[
        A,B\in\mathcal S_r^{m,n}
\]
with
\[
        A\sigma B\notin\mathcal S_{2r-1}^{m,n}.
\]
\end{proof}

\section{Applications to Choi cones of quantum channels}

The preceding results can be reformulated naturally in the language of
quantum maps.  Under the Choi--Jamiołkowski correspondence, bipartite
positive semidefinite matrices are precisely the Choi matrices of
completely positive maps, while the separable, PPT, and Schmidt-number
cones become the Choi cones of entanglement-breaking, PPT, and
\(k\)-superpositive maps, respectively.  We now translate the preceding cone-rigidity results into statements
about completely positive maps via the Choi--Jamio{\l}kowski
correspondence.: Kubo--Ando means are
compatible with complete positivity at the Choi level, but their
non-arithmetic curvature is incompatible with the finer cone structures
arising from entanglement and positive-map constraints.  We now formulate
this consequence for completely positive maps and quantum channels.

Let \(\sigma\) be a Kubo--Ando mean.  If
\(\Phi,\Psi:\M_m\to \M_n\) are completely positive maps, we   consider  the $\sigma$-mean of two completely positive maps $\Phi$ and $\Psi$ as the unique map whose Choi matrix satisfies $C_{\Phi\sigma\Psi}
        =
        C_\Phi \sigma C_\Psi$. Recently, some  framework for extending Kubo–Ando means to completely positive maps via Arveson's Radon–Nikodym derivatives was also introduced \cite{kian2026,okayasu2026}, see also    \cite{FMVW}.
Since \( C_\Phi \sigma C_\Psi\geq0\), the map
\(\Phi\sigma\Psi\) is again completely positive.

If \(\Phi\) and \(\Psi\) are quantum channels, however, the map
\(\Phi\sigma\Psi\) need not be trace preserving, since in
general
\[
        \operatorname{Tr}_n\bigl( C_\Phi \sigma C_\Psi\bigr)
        \neq I_m .
\]
Thus, in order to obtain a binary operation on channels, we introduce a
normalization of Choi matrices.

Let \(C\in \M_m\otimes \M_n\) be positive definite.  Define
\[
        \mathcal N(C)
        =
        \bigl(
        (\operatorname{Tr}_n C)^{-1/2}\otimes I_n
        \bigr)
        C
        \bigl(
        (\operatorname{Tr}_n C)^{-1/2}\otimes I_n
        \bigr).
\]
Then
\[
        \operatorname{Tr}_n\mathcal N(C)=I_m.
\]
Accordingly, for full-rank quantum channels \(\Phi,\Psi:\M_m\to \M_n\), we
define their normalized Choi-level \(\sigma\)-mean by
\[
        J(\Phi\widehat{\sigma}\Psi)
        =
        \mathcal N\bigl(J(\Phi)\sigma J(\Psi)\bigr).
\]

The operation \(\widehat{\sigma}\) is therefore a Choi-level nonlinear
mean of quantum channels.  We note that this should not be confused with ordinary
composition of maps; because  it is obtained by applying the Kubo--Ando
mean to the corresponding Choi matrices and then restoring trace
preservation by a local normalization.

\begin{lemma}\label{lem-choi-normal}
Let \(C\in \M_m\otimes \M_n\) be positive definite. Then
\[
        \operatorname{Tr}_n\mathcal N(C)=I_m.
\]
Moreover, for every \(1\leq k\leq \min\{m,n\}\),
\[
        C\in\mathcal S_k^{m,n}
        \quad\Longleftrightarrow\quad
        \mathcal N(C)\in\mathcal S_k^{m,n}.
\]
Also,
\[
        C\in\mathcal {PPT}^{m,n}
        \quad\Longleftrightarrow\quad
        \mathcal N(C)\in\mathcal {PPT}^{m,n}.
\]
\end{lemma}

\begin{proof}
Set
\[
        R=(\operatorname{Tr}_n C)^{-1/2}.
\]
Since \(C>0\), we have \(\operatorname{Tr}_n C>0\), and hence \(R\) is
invertible.  By definition,
\[
        \mathcal N(C)=(R\otimes I_n)C(R\otimes I_n).
\]
Taking the partial trace over the second tensor factor gives
\[
        \operatorname{Tr}_n\mathcal N(C)
        =
        R(\operatorname{Tr}_n C)R
        =
        I_m.
\]
We next prove the assertion about Schmidt number. If $\xi \in \mathbb C^m\otimes \mathbb C^n$ has   Schmidt rank at most \(k\), then \(\xi\) admits a decomposition
\[
    \xi = \sum_{j=1}^k x_j\otimes y_j
\]
for some vectors \(x_j\in\mathbb C^m\) and
\(y_j\in\mathbb C^n\). Applying \(R\otimes I_n\) gives
\[
    (R\otimes I_n)\xi
    =
    \sum_{j=1}^k (Rx_j)\otimes y_j,
\]
which is again a sum of at most \(k\) elementary tensors.  Hence
\[
    \operatorname{SR}\big((R\otimes I_n)\xi\big)\le k.
\]
Since \(R\) is invertible, the same argument applied to
\(R^{-1}\otimes I_n\) shows the converse implication, so
\(R\otimes I_n\) preserves Schmidt rank exactly.

Now let \(C\in \mathcal S_k^{m,n}\).  By definition,
\[
    C = \sum_\ell \xi_\ell\xi_\ell^*
\]
where each \(\xi_\ell\) has Schmidt rank at most \(k\).  Therefore
\[
    (R\otimes I_n)C(R\otimes I_n)
    =
    \sum_\ell
    \big((R\otimes I_n)\xi_\ell\big)
    \big((R\otimes I_n)\xi_\ell\big)^*.
\]
Since each vector \((R\otimes I_n)\xi_\ell\) also has Schmidt rank at
most \(k\), it follows that
\[
    (R\otimes I_n)C(R\otimes I_n)\in \mathcal S_k^{m,n}.
\]
Applying the same argument to \(R^{-1}\otimes I_n\) yields the reverse
inclusion, and hence local invertible congruences preserve the cone
\(\mathcal S_k^{m,n}\).

Finally, we prove the PPT assertion.  Since the partial transpose acts on
the second tensor factor, while \(R\) acts on the first tensor factor, we
have
\[
\begin{aligned}
        (\operatorname{id}_m\otimes T)(\mathcal N(C))
        &=
        (\operatorname{id}_m\otimes T)
        \bigl((R\otimes I_n)C(R\otimes I_n)\bigr)  \\
        &=
        (R\otimes I_n)
        (\operatorname{id}_m\otimes T)(C)
        (R\otimes I_n).
\end{aligned}
\]
Again this is an invertible congruence.  Hence
\[
        (\operatorname{id}_m\otimes T)(C)\geq0
\]
if and only if
\[
        (\operatorname{id}_m\otimes T)(\mathcal N(C))\geq0.
\]
Thus
\[
        C\in\mathcal {PPT}^{m,n}
        \quad\Longleftrightarrow\quad
        \mathcal N(C)\in\mathcal {PPT}^{m,n}.
\]
\end{proof}

We recall the corresponding terminology for completely positive maps.
A completely positive map \(\Phi:\M_m\to \M_n\) is called
\begin{enumerate}
  \item entanglement-breaking if $C_\Phi\in\mathcal S_1^{m,n}\);
  \item positive partial transpose (PPT)  if \(C_\Phi\in\mathcal {PPT}^{m,n}\);
  \item \(k\)-superpositive if \(C_\Phi\in\mathcal S_k^{m,n}\).
\end{enumerate}

  More generally, if
\(\mathcal C\subseteq \M_m\otimes \M_n\) is a cone, we write
\[
        \operatorname{CP}_{\mathcal C}(m,n)
        =
        \{\Phi:\M_m\to \M_n\text{ completely positive}:
        C_\Phi\in\mathcal C\}.
\]
Thus \(\operatorname{CP}_{\mathcal S_1}(m,n)\) is the cone of
entanglement-breaking completely positive maps, while
\(\operatorname{CP}_{\mathcal S_k}(m,n)\) is the cone of
\(k\)-superpositive maps.

\begin{corollary} \label{thm-cone}
Let \(m,n\geq2\), and let \(\sigma\) be a Kubo--Ando mean.  Let
\(\mathcal C\subseteq \M_m\otimes \M_n\) be a convex cone satisfying
\[
        \mathcal S_1^{m,n}
        \subseteq
        \mathcal C
        \subseteq
        \mathcal{PPT}^{m,n}.
\]
Then the following are equivalent:
\begin{enumerate}
        \item \(\operatorname{CP}_{\mathcal C}(m,n)\) is closed under the Choi-level mean $\sigma$;
    \item \(\sigma\) is a weighted arithmetic mean.
    \end{enumerate}
In particular, a Kubo--Ando mean preserves entanglement-breaking
completely positive maps under \(\sigma\) if and only if it
is arithmetic.
\end{corollary}

\begin{proof}
Suppose first that \(\sigma\) is a weighted arithmetic mean.  Then, for
some \(t\in[0,1]\),
\[
        A\sigma B=(1-t)A+tB.
\]
If \(\Phi,\Psi\in\operatorname{CP}_{\mathcal C}(m,n)\), then $C_\Phi,C_\Psi \in\mathcal C$.
Since \(\mathcal C\) is convex,
\[
        C_{\Phi\sigma\Psi}=(1-t)C_\Phi+tC_\Psi\in\mathcal C.
\]
Hence
\[
        \Phi\sigma\Psi
        \in
        \operatorname{CP}_{\mathcal C}(m,n).
\]
Conversely, suppose that \(\sigma\) is not arithmetic.  By
Theorem~\ref{thm-main3}, there exist $A,B\in\mathcal S_1^{m,n}$
such that
\[
        A\sigma B\notin\mathcal {PPT}^{m,n}.
\]
Let \(\Phi,\Psi:\M_m\to \M_n\) be the completely positive maps with Choi matrices $C_\Phi=A$ and $C_\Psi=B$.
Since \(A,B\in\mathcal S_1^{m,n}\subseteq\mathcal C\), we have
\[
        \Phi,\Psi\in\operatorname{CP}_{\mathcal C}(m,n).
\]
However,
\[
         C_{\Phi\sigma\Psi}=
                C_\Phi \sigma C_\Psi
        =
        A\sigma B
        \notin
        \mathcal {PPT}^{m,n}.
\]
Since
\[
        \mathcal C\subseteq\mathcal {PPT}^{m,n},
\]
it follows that
\[
        C_{\Phi\sigma\Psi} \notin\mathcal C.
\]
Therefore
\[
        \Phi\sigma\Psi
        \notin
        \operatorname{CP}_{\mathcal C}(m,n).
\]
Thus \(\operatorname{CP}_{\mathcal C}(m,n)\) is not closed under
\(\sigma\).  This proves the equivalence.
\end{proof}

Thus the cone-rigidity theorem has an immediate map-theoretic
interpretation: among Kubo--Ando means, the arithmetic means are exactly
those whose Choi-level action preserves the completely positive maps
associated with any cone between separability and PPT.  In particular,
non-arithmetic Kubo--Ando means may take two entanglement-breaking
completely positive maps to a map that is not even PPT.

\begin{theorem}\label{thm-ent-rig}
Let \(\sigma\) be a Kubo--Ando mean.  Then the following are equivalent:
\begin{enumerate}
\item The class of full-rank entanglement-breaking qubit channels is
    closed under the normalized Choi-level mean \(\widehat{\sigma}\).
    \item \(\sigma\) is a weighted arithmetic mean.
    \end{enumerate}
Equivalently, if \(\sigma\) is not arithmetic, then there exist full-rank
entanglement-breaking qubit channels \(\Phi,\Psi:\M_2\to \M_2\) such that
\[
        \Phi\widehat{\sigma}\Psi
\]
is not PPT, and hence is not entanglement breaking.
\end{theorem}

\begin{proof}
First suppose  that \(\sigma\) is a weighted arithmetic mean.  Then, $A\sigma B=(1-t)A+tB$ for
some \(t\in[0,1]\). Let \(\Phi,\Psi:\M_2\to \M_2\) be entanglement-breaking channels.  Then
 $C_\Phi,C_\Psi\in \mathcal S_1^{2,2}$
and
\[
        \operatorname{Tr}_2 C_\Phi
        =
        \operatorname{Tr}_2 C_\Psi
        =
        I_2 .
\]
Hence
\[
        C_\Phi\sigma C_\Psi
        =
        (1-t)C_\Phi+t C_\Psi
\]
is again separable and has partial trace \(I_2\).  Thus the normalization
is unnecessary, and
\[
        C_{\Phi\widehat{\sigma}\Psi}
        =
        (1-t)C_\Phi+tC_\Psi
        \in \mathcal S_1^{2,2}.
\]
Therefore \(\Phi\widehat{\sigma}\Psi\) is entanglement breaking.

Conversely, suppose that \(\sigma\) is not arithmetic. By the two-qubit construction in Theorem~\ref{thm-main1}, there exist
full-rank separable matrices
\[
    A,B\in \operatorname{int}\mathcal S_1^{2,2}
\]
which, in addition, we have
\[
    \operatorname{Tr}_2 A=\alpha I_2,
    \qquad
    \operatorname{Tr}_2 B=\beta I_2
\]
By setting
\[
        \widetilde A=\alpha^{-1}A,
        \qquad
        \widetilde B=\beta^{-1}B,
\]
we  still have
\[
        \widetilde A\sigma \widetilde B
        \notin
        \mathcal {PPT}^{2,2}.
\]
Moreover,
\[
        \operatorname{Tr}_2\widetilde A
        =
        \operatorname{Tr}_2\widetilde B
        =
        I_2 .
\]
Thus \(\widetilde A\) and \(\widetilde B\) are Choi matrices of full-rank
entanglement-breaking qubit channels, say
\[
        C_\Phi=\widetilde A,
        \qquad
      C_\Psi=\widetilde B .
\]
Then
\[
        C_\Phi\sigma C_\Psi
        =
        \widetilde A\sigma\widetilde B
        \notin
        \mathcal {PPT}^{2,2}.
\]
By Lemma~\ref{lem-choi-normal}, Choi
normalization preserves PPT membership.  Therefore
\[
        C_{\Phi\widehat{\sigma}\Psi}
        =
        \mathcal N\bigl(C_\Phi\sigma C_\Psi\bigr)
        \notin
        \mathcal {PPT}^{2,2}.
\]
Hence \(\Phi\widehat{\sigma}\Psi\) is not PPT and, in particular, is not
entanglement breaking.
\end{proof}

This gives a channel-theoretic interpretation of the cone rigidity
phenomenon.  Kubo--Ando means always preserve complete positivity at the
level of Choi matrices, but the preceding theorem shows that this is the
end of the general preservation principle: except for arithmetic means,
the normalized Choi-level operation fails to preserve even
entanglement-breaking qubit channels.  Thus convex mixing is the only
Kubo--Ando averaging procedure compatible with entanglement-breaking
structure.

\medskip

\noindent \textit{Conflict of Interest Statement.}  There is no conflict of interest.
\medskip

\noindent \textit{Ethical Statement.}  Not applicable. This research did not involve human participants, personal data, or animals.

\medskip
\noindent \textit{Informed Consent.} Not applicable.

\medskip
\noindent \textit{Funding Statement.} No funding was received for conducting this study.

\medskip
\noindent\textit{Data Availability Statement.} Data sharing not applicable to this article as no datasets were generated or analyzed during the current study.


\medskip
\begin{thebibliography}{99}


\bibitem{Choi1975}
M.-D. Choi,
\textit{Completely positive linear maps on complex matrices},
Linear Algebra Appl. \textbf{10} (1975), 285--290.
\bibitem{FMVW} P. E. Frenkel, M. Mosonyi, P. Vrana, and M. Weiner, \textit{Error bounds for composite quantum hypothesis testing and a new
characterization of the weighted Kubo--Ando geometric means}, (2025), arXiv:2503.13379.

\bibitem{FMPS} T. Furuta, J. Mi\'{c}i\'{c} Hot, J. Pe\v{c}ari\'{c} and Y. Seo, {\it Mond-Pe\v{c}ari\'{c} Method in Operator Inequalities}, Monographs in Inequalities 1, Element, Zagreb, 2005.

\bibitem{Horodecki1996}
M. Horodecki, P. Horodecki, and R. Horodecki,
\textit{Separability of mixed states: necessary and sufficient conditions},
Phys. Lett. A \textbf{223} (1996), 1--8.

\bibitem{HorodeckiShorRuskai2003}
M. Horodecki, P. W. Shor, and M. B. Ruskai,
\textit{Entanglement breaking channels},
Rev. Math. Phys. \textbf{15} (2003), 629--641.




\bibitem{Jamiolkowski1972}
A. Jamio{\l}kowski,
\textit{Linear transformations which preserve trace and positive semidefiniteness of operators},
Rep. Math. Phys. \textbf{3} (1972), 275--278.

\bibitem{kian-rc2025} M. Kian,  \textit{ A hierarchy of entanglement cones via rank-constrained $C^*$-convex hulls}, (2025) arXiv:2512.05560.

  \bibitem{kian2026}  M. Kian,  \textit{ Relative Kubo-Ando means of completely positive maps},    (2026) arXiv:2605.11701.

 \bibitem{KMS}   M. Kian   M. S. Moslehian,  and Y.  Seo,    \textit{ Variants of Ando–Hiai type inequalities for deformed means and applications},  Glasg.  Math. J.   \textbf{63} (3) (2021), 622--639.


\bibitem{KuboAndo1980}
F. Kubo and T. Ando,
\textit{Means of positive linear operators},
Math. Ann. \textbf{246} (1980), 205--224.

\bibitem{okayasu2026} R. Okayasu,  \textit{Geometric means and Lebesgue-type decomposition of completely positive maps}, (2026)  arXiv:2605.06019v1.

\bibitem{Peres1996}
A. Peres, \textit{Separability criterion for density matrices},
Phys. Rev. Lett. \textbf{77} (1996), 1413--1415.


\bibitem{TerhalHorodecki2000}
B. M. Terhal and P. Horodecki,
\textit{Schmidt number for density matrices},
Phys. Rev. A \textbf{61} (2000), 040301.


\end{thebibliography}
\end{document}